\documentclass[A4paper,11pt]{article}

\usepackage{amsmath,amssymb,amsfonts,graphicx}
\newtheorem{theorem}{Theorem}[section]
\newtheorem{proposition}[theorem]{Proposition}

\newtheorem{definition}[theorem]{Definition}

\newtheorem{notation}[theorem]{Notation}

\newenvironment{proof}{\mbox{\bf Proof.}}{\mbox{$\dashv$}\bigskip}
\begin{document}
\begin{center}
{\Large\bf A Constructive Epistemic    Logic with Public Announcement}\\
{\bf(Non-Predetermined Possibilities) }\\
{\footnotesize {[DRAFT Feb.2013]}}

 \vspace{.25in}
{\bf   Rasoul Ramezanian}\\
Department of Mathematics
Sharif University of Technology,\\
P. O. Box 11365-9415, Tehran, Iran\\
ramezanian@sharif.edu
\end{center}
\begin{abstract}
\noindent We argue that the notion of epistemic \emph{possible
worlds} in constructivism (intuitionism) is not as the same as it
is in classic view, and there are possibilities, called
non-predetermined worlds, which are ignored in (classic) Epistemic
Logic. Regarding non-predetermined possibilities, we propose a
    constructive epistemic logic  and prove soundness and
completeness theorems for it. We extend the proposed logic by
adding a public announcement operator. To declare the significance
of our work, we formulate
  the well-known Surprise Exam Paradox, $\mathbf{SEP}$, via the proposed
  constructive   epistemic logic   and then put forward a solution for the paradox.
  We clarify
  that the puzzle in the $\mathbf{SEP}$ is because of
 students'(wrong) assumption   that the day of the exam is
  necessarily predetermined.

\vspace*{1 cm}
\end{abstract}

 In (classic) epistemic logic (see~\cite{kn:dit3}), the   knowledge of an agent is
modelled through   two fundamental  notions
\begin{itemize}
\item[1-] \emph{possible worlds} (states) and \item[2-]
\emph{Indistinguishability}. \end{itemize}
 An agent knows  some fact if it is true in all \emph{possible worlds} that the agent
 cannot \emph{distinguish} them from the actual world.

 To propose a Constructive  Epistemic Logic, we first emphasize
 that the notion of \emph{possible worlds} in intuitionism (constructivism) is not
 as the same as it is in classic view. It is our main   idea
 that leads us to introduce an Epistemic logic from intuitionistic
 (constructive) point of view.

Our proposed constructive  epistemic logic is not much different
with (classic) epistemic logic except that it admits a new kind of
possible worlds called non-predetermined worlds where the facts of
the worlds have not necessarily been  determined already!

 Suppose
I announce on Facebook that I like one  and only one of the days
of the next week, and   on each day, I will write on my Facebook
at 10am and  announce that whether I like the day or not. Also,
suppose that it is Monday 10am now, and I am going to announce my
idea about the day. What is the actual world that I am present in?
I can announce  $p:$ ``I like Monday" and also I can announce
$\neg p$: ``I do not like Monday". That is, in the actual world
that I live, it is not  determined neither $p$ nor $\neg p$ yet.
Kripke models which are defined as semantics of (classic)
epistemic logic cannot model the notion of predetermination. For
each (classic) epistemic state $(M,s)$   and each atomic formula
$q$, we have either $(M,s)\models q$ or $(M,s)\models \neg q$. Let
$(N,t)$ be a (classic) epistemic state which describes my state on
Monday. Then   either $(N,t)\models p$ or $(N,t)\models \neg p$,
(and certainly not both of them). If $(N,t)\models p$, then I
cannot announce ``I do not like Monday", whereas it is up to my
free will, and what I like is not necessary predetermined. If
$(N,t)\models \neg p$, then I cannot announce ``I   like Monday",
and I am forced to write ``I like Monday"!
\begin{quote}On Monday at 10am, both $p$ and $\neg p$ are
announce-able and a suitable \emph{epistemic possible world} to
describe my state at Monday, should not satisfy neither  $p$ nor
$\neg p$.
\end{quote}
Therefore, it is not possible to describe the actual state that I
have on Monday via classic epistemic  possible worlds~\footnote{If
we want to describe the actual state classically then we need to
consider Temporal concepts and thinking of temporal epistemic
logics. However, in our work, we show that we can describe
non-predetermined cases without using temporal modals and
operators.}. We need a kind of possible worlds which at them
neither $p$ nor $\neg p$ is necessarily predetermined. We use Beth
models (a semantic class for constructive  logic,
see~\cite{kn:dalen}) to obtain this aim.

$\qquad$

\noindent  The paper is organized as follows:
\begin{itemize}
\item[] In section~\ref{IPW}, we discuss the notion of
possibility,  and argue that non-predetermination should be
considered as a new possibility. We use Beth models
(see\cite{kn:dalen}) to describe non-predetermination
possibilities, and then as semantics for our constructive
epistemic logic, we introduce a kind of Kripke models, called
Beth-Kripke models, where each possible world  is a Beth model
(instead of a valuation to atomic formula). We then provide an
axiomatization system and prove soundness and completeness
theorem.

\item[] In  section~\ref{IPAL}, we extend our proposed logic, by
adding the public announcement operator. The public announcement
operator is defined in the way that for those formulas, say
$\varphi$ which their value are not  yet  determined in the actual
world (say $(M,s)$), both $\varphi$ and $\neg\varphi$ are
announce-able (in other words, we have both $(M,s)\models \langle
\varphi\rangle \top$ and $(M,s)\models \langle \neg\varphi\rangle
\top$).

\item[]  The surprise exam paradox, $\mathbf{SEP}$
(see~\cite{kn:Conner,kn:Quine,kn:Shaw}), was formulated via
classic epistemic logic in different
ways~\cite{kn:Bink,kn:Quine,kn:Sore,kn:william}. Also the paradox
was formulated in constructive analysis~\cite{kn:RamAP}. In
section~\ref{SEP}, we formally model $\mathbf{SEP}$ in the
proposed  constructive  epistemic logic. Then, regarding
non-predetermined worlds, we put forward a solution for the
paradox.
\end{itemize}
\section{ Constructive Possible Worlds}\label{IPW}
The semantics of classic propositional logic is introduced via
valuations of atomic formulas by $True$ or $False$. Consider two
atomic propositional formulas $p$ and $q$. Classically, there are
exactly  four different possible worlds (valuations) for $p$ and
$q$ as follows:
\begin{itemize}
\item[1-] $p=True$, $q=True$,

\item[2-] $p=True$, $q=False$,

\item[3-] $p= False$, $q=True$,

\item[4-] $p=False$, $q=False$.
\end{itemize}
However, the semantics of constructive (intuitionistic)
propositional logic is not the same semantics of classic
propositional logic. The semantics of constructive propositional
logic is formally introduced by Kripke models or Beth models. In
this paper, we consider Beth models~(see~\cite{kn:dalen}, and
chapter~13 of~\cite{kn:TD}). A Beth model is a triple
$\Theta=\langle Q,\leq, F\rangle$, where $\langle Q,\leq\rangle$
is a partially ordered set with the following condition that there
exists a node $\alpha\in Q$, called root, such that for all
$\beta,\gamma\in Q$, $\alpha\leq \beta$, $\alpha\leq
\gamma$~\footnote{This condition is extra, and we consider it here
for convenience.}.

\noindent $F$ is mapping assigning atomic formulas to elements of
$Q$. More precisely, let $\mathbf{AT}$ be the set all atomic
propositional formulas. Then $F: Q\rightarrow 2^{\mathbf{AT}}$ is
a function subject to the following condition: for all
$\alpha,\beta\in Q$, if $\alpha\leq\beta$ then $F(\alpha)\subseteq
F(\beta)$.

\begin{notation} Given a Beth model $\Theta=\langle Q,\leq, F\rangle$, instead
of writing $\alpha\in Q$,  for simplicity, we write $\alpha\in
\Theta$.
\end{notation}

\noindent A path $P$ through a node $\alpha\in \Theta$ is a
maximal linearly ordered subset of $\Theta$ containing $\alpha$. A
bar $B$ for a node  $\alpha\in \Theta$ is a subset of $\Theta$
with the property that each path through $\alpha$ intersects it.

The satisfaction relation $\Vdash\subseteq \Theta\times
SENT(\mathbf{AT})$ (where $SENT(\mathbf{AT})$ is the set of all
propositional formulas over $\textbf{AT}$), is defined
inductively,
\begin{itemize}
\item[] $\alpha\Vdash p$ iff there is a bar $B$ for $\alpha$  such
that for each $\beta\in B$, $p\in F(\beta)$, (for atomic $p\in
\mathbf{AT}$).

\item[] $\alpha\Vdash A\wedge B$ iff $\alpha\Vdash A$ and
$\alpha\Vdash B$.

\item[] $\alpha\Vdash A\vee B$ iff there is a bar $B$ for
$\alpha$, such that  for each $\beta\in B$, $\beta\Vdash A$ or
$\beta\Vdash B$.

\item[] $\alpha \Vdash A\rightarrow B$ iff for each
$\beta\geq\alpha$, $\beta\Vdash A$ then $\beta\Vdash B$.

\item[] $\alpha\Vdash \neg A$ iff for each $\beta\geq\alpha$,
$\beta\not\Vdash A$.
\end{itemize}
We say two Beth models $(\Theta_1,\alpha_1)$ and
$(\Theta_2,\alpha_2)$ are equivalent whenever they satisfy the
same propositional formulas.

\begin{theorem} \label{beth1} Given a Beth model $(\Theta,\alpha)$, for any proportional formula $A$,
\begin{itemize}
\item[a.] $\alpha\Vdash A$ iff there is a bar $B$ for $\alpha$
such that for all $\beta\in B$, $\beta\Vdash A$.

\item[b.] $\alpha\not\Vdash A$ iff there is a path $P$ through
$\alpha$ such that for each $\beta\in P$, $\beta\not\Vdash A$.

\item[c.] $\alpha\leq\beta$ and $\alpha\Vdash A$ then $\beta\Vdash
A$.
\end{itemize}
\end{theorem}\begin{proof} See~\cite{kn:dalen}. \end{proof}

Let $\Gamma$ be a set of propositional formulas. By
$\Gamma\vdash_i A$ we mean $A$ is derivable in constructive
propositional logic. By $\Gamma\models A$, we mean $A$ is
satisfied  in all Beth models which satisfy all formulas in
$\Gamma$.

\begin{theorem}\label{SCT} Soundness and Completeness Theorem: $\Gamma\vdash_i A$ iff
$\Gamma\models A$.

\end{theorem}\begin{proof}See~\cite{kn:dalen}. \end{proof}

In the beginning of the section, we mentioned that there are four
different possible worlds for two atomic formulas $p$ and $q$ in
classic  view. In constructive propositional logic, regarding Beth
models (instead of valuations), the number of possible worlds are
more than four.

\noindent Let $\Gamma=\{p\vee q, \neg(p\wedge q)\}$. Classically,
two possible worlds are considerable for $\Gamma$
\begin{itemize}
\item[c1-] $p=True$, $q=False$, and

\item[c2-] $p=False$, $q=True$.
\end{itemize}
But there are three non-equivalent Beth models for $\Gamma$,
\begin{itemize}
\item[i1-] $(\Theta,\alpha)$, where $\Theta=\{\alpha\}$, and
$F(\alpha)=\{p\}$,

 \item[i2-] $(\Theta',\alpha')$, where
$\Theta'=\{\alpha'\}$, and $F(\alpha')=\{q\}$, and

 \item[i3-] $(\Theta'',\alpha'')$, where
$\Theta''=\{\alpha'', \beta'',\gamma''\}$,
$F(\alpha'')=\emptyset$, $F(\beta'')=\{p\}$, $F(\gamma'')=\{q\}$,
$\alpha''\leq \beta''$, and $\alpha''\leq \gamma''$.

\end{itemize}

%\begin{center}
%\scalebox{.6}{\includegraphics{pic1.jpg}}

%$(\Theta'',\alpha'')$
%\end{center}
The two cases $i1$ and $i2$ are the same classical cases $c1$ and
$c2$, but the third one, $i3$, is new. The possible world
$(\Theta'',\alpha'')$ is a situation where formulas $p\vee q,
\neg(p\wedge q)$ holds true but neither $p$  nor $q$ are
\emph{predetermined}. The possibility of non-predetermination is
regarded in Beth models, whereas classically, it is presupposed
that valuation of atomic formulas are predetermined already.

Recalling the Surprise Exam Paradox (see section~\ref{SEP}),
suppose $p$ stands for "Tomorrow, the teacher will take  the
exam", and $q$ stands for " the teacher will take  the exam the
day after tomorrow"

\begin{itemize} \item[1-] The Beth model $(\Theta,\alpha)$ represents
the possible world where it is already determined that teacher
will take the exam tomorrow.

\item[2-] The Beth model $(\Theta',\alpha')$ represents the
possible world where it is already determined that teacher will
take the exam the day after tomorrow.

\item[3-] The Beth model $(\Theta'',\alpha'')$ represents the
possible world where the teacher has not already decided whether
take the exam   tomorrow or one day later. Classically, we cannot
represent the third possibility.
\end{itemize}

One may assume the root of the Beth model $(\Theta'',\alpha'')$ as
the current state of the possible world. Other nodes with respect
to partial order relation are future nodes of the current
state~\footnote{The reader may note that to regard future, we can
also argue in terms of temporal logic. But constructive
(intuitionistic) logic, without having temporal modal operators,
in some sense, considers this case. Therefore, avoiding extra
modal operators, we propose a  constructive epistemic logic.}.
Future is not predetermined and it is the teacher who determines
it later by his \emph{free will}. It is up to free will of the
teacher to take the exam tomorrow or not, and neither the value of
$p$ nor the value of $\neg p$ is not determined already.
 Beth models
help us describe \emph{non-predetermination} as a new possibility.
The \emph{non-predetermination} is ignored in (classic) epistemic
logic.
\begin{definition} Let $\mathbf{AT}$ be a non-empty set of propositional
variables, and $\mathcal{A}$ be a set of agents. The language
$L(\mathcal{A},\mathbf{AT})$ is the smallest superset of
$\mathbf{AT}$ such that
\begin{center}
if $\varphi,\psi\in L(\mathcal{A},\mathbf{AT})$ then $\neg
\varphi,\ (\varphi\wedge\psi), (\varphi\vee \psi),
\varphi\rightarrow\psi, \ K_i\varphi\in
L(\mathcal{A},\mathbf{AT})$,
\end{center}
for $i\in \mathcal{A}$.
\end{definition}
For $i\in \mathcal{A}$, $K_i\varphi$ has to be read as `agent $i$
knows $\varphi$".
\begin{definition}\label{applic}
A Beth-Kripke model $M$ is a tuple $M=\langle S,(\sim_i)_{i\in
\mathcal{A}}\rangle$, where $S$ is a non-empty set of Beth models
over $\mathbf{AT}$ as possible worlds, (each   $s\in S$   is a
pointed Beth model $(\Theta_s, \alpha_s)$ where $\alpha_s$ is the
root of $\Theta_s$), and each $\sim_i$ is a binary accessibility
relation between worlds.
\end{definition}
Let $M=\langle S,(\sim_i)_{i\in \mathcal{A}}\rangle$ be a
Beth-Kripke model, and $\Theta$ be an arbitrary Beth model of $M$
(i.e., $\Theta=\Theta_s$ for some $s\in S$). We define the
satisfaction relation $\Vdash_M\subseteq \Theta\times
L(\mathcal{A},\mathbf{AT})$ as follows:

\begin{itemize}
\item[-] for all $\varphi\in SENT(\mathbf{AT})$, for each
$\alpha\in \Theta$, $\alpha\Vdash_M \varphi$ iff  $\alpha\Vdash
\varphi$.

\item[-] For each $\alpha\in \Theta$, for each agent $i\in
\mathcal{A}$,  $\alpha \Vdash_M K_i\varphi$ iff for all $\Omega
\in S$, if  $\Theta\sim_i \Omega$, then for all $\beta\in \Omega$,
$\beta\Vdash_M \varphi$.

\item[-]For each $\alpha\in \Theta$, $\alpha\Vdash_M \varphi\wedge
\psi$ iff $\alpha\Vdash_M \varphi$ and $\alpha\Vdash_M \psi$.

\item[-] For each $\alpha\in \Theta$, $\alpha\Vdash_M \varphi\vee
\psi$ iff there is a bar $B\subseteq \Theta$  for $\alpha$, such
that for each $\beta\in B$, $\beta\Vdash_M \varphi$ or
$\beta\Vdash_M \psi$.

\item[-] For each $\alpha\in \Theta$, $\alpha \Vdash_M
\varphi\rightarrow \psi$ iff for each $\beta\geq_\Theta\alpha$,
$\beta\Vdash_M \varphi$ then $\beta\Vdash_M \psi$.

\item[-] $\alpha\Vdash_M \neg \varphi$ iff for each
$\beta\geq_\Theta\alpha$, $\beta\not\Vdash_M \varphi$.

\end{itemize}
%\end{definition}

\begin{theorem}\label{bethk}
Let $M=\langle S,(\sim_i)_{i\in \mathcal{A}}\rangle$ be a
Beth-Kripke model, and $\Theta$ be an arbitrary Beth model of $M$.
For each $\alpha\in \Theta$, and $\varphi\in
L(\mathcal{A},\mathbf{AT})$,  we have
\begin{itemize}
\item[a.] $\alpha\Vdash_M  \varphi$ iff there is a bar $B$ for
$\alpha$ such that for all $\beta\in B$, $\beta\Vdash_M \varphi$.

\item[b.] $\alpha\not\Vdash_M \varphi$ iff there is a path $P$
through $\alpha$ such that for each $\beta\in P$,
$\beta\not\Vdash_M \varphi$.

\item[c.] $\alpha\leq\beta$ and $\alpha\Vdash_M \varphi$ then
$\beta\Vdash_M \varphi$.

\item[d.] $\alpha\Vdash_M K_i \varphi$ then for all $\beta\in
\Theta$, $\beta\Vdash_M K_i\varphi$.
\end{itemize}
\end{theorem}\begin{proof} It is straightforward.\end{proof}

\begin{definition}\label{sat1}
Let $\varphi\in L(\mathcal{A},\mathbf{AT})$ and $M=\langle
S,(\sim_i)_{i\in \mathcal{A}}\rangle$ be a Beth-Kripke model. We
say $(M,s)$ satisfies $\varphi$, denoted by $(M,s)\models
\varphi$, whenever the root of $\Theta_s$ satisfies the formula
$\varphi$ regarding the model $M$, i.e., $\alpha_s\Vdash_M
\varphi$.

\end{definition}

\begin{theorem}
Let $M=\langle S,(\sim_i)_{i\in \mathcal{A}}\rangle$, $s\in S$,
and $i\in \mathcal{A}$ be an arbitrary agent. Also let
$\Gamma=\{\varphi\in SENT(\mathbf{AT})\mid (M,s)\models K_i\varphi
\}$. Then $\Gamma$ is a constructive (intuitionistic)
propositional closed theory. That is, for every $\psi\in
SENT(\mathbf{AT})$, $\psi\in \Gamma$, iff $\psi$ is derivable from
$\Gamma$ in constructive propositional logic, $\Gamma\vdash_i
\psi$.
\end{theorem}\begin{proof} Suppose $\Gamma\vdash_i \varphi$ in constructive (intuitionistic)  logic. Then by soundness
and completeness theorem~\ref{SCT}, $\varphi$ is satisfied in all
Beth models which satisfy all formulas in $\Gamma$. Hence
$\varphi$ is satisfied in all Beth model $(\Theta_t,\alpha_t)$ in
model $M$, which $t\sim_i s$. Therefore, by definition of
satisfaction for knowledge, we have $(M,s)\models K_i\varphi$, and
we are done.
\end{proof}
Above theorem declares that
\begin{quote} knowing a fact in constructive point of view is to have a
proof for the fact.
\end{quote}

\begin{theorem} Let $M=\langle S,(\sim_i)_{i\in \mathcal{A}}\rangle$, $s\in S$,
 $i\in \mathcal{A}$, and $\varphi\in L(\mathcal{A},\mathbf{AT})$.
 Then $(M,s)\models K_i\varphi\vee \neg K_i\varphi$.
\end{theorem}\begin{proof} By definition~\ref{sat1}, $(M,s)\models K_i\varphi\vee \neg K_i\varphi$
is equivalent to $\alpha_s\Vdash_M K_i\varphi\vee \neg
K_i\varphi$. Note that, according to item d) of
theorem~\ref{bethk}, for all $\beta\in \Theta_s$, if
$\beta\Vdash_M K_i\varphi$ then for all $\delta\in \Theta_s$, we
have $\delta\Vdash_M K_i\varphi$. Therefore,  Either none of the
nodes of the Beth model $\Theta$ satisfies $K_i\varphi$, and
consequently by  definition~\ref{sat1}, we have $\alpha_s\Vdash_M
\neg K_i\varphi$. Or all nodes of the Beth model $\Theta$
satisfies $K_i\varphi$, and consequently, $\alpha_s\Vdash_M
K_i\varphi$. Thus, we have $\alpha_s\Vdash_M K_i\varphi\vee \neg
K_i\varphi$.
\end{proof}
The above theorem says, knowing is decidable. Intuitionistic
(constructive) propositional logic, $\mathbf{IPC}$
(see~\cite{kn:TD}, chapter~2), is effectively decidable. That is,
it is decidable that whether a formula is a theorem or not. As the
notion of knowledge in intuitionism is considered equal to having
evidence (proof), it is plausible that for an appropriate
constructive  epistemic logic knowing is decidable. The two above
theorems justifies why we call our proposed epistemic logic as a
constructive epistemic logic.

\subsection {Axiomatization}
In this part, we introduce an axiomatization system dented by
$\mathbf{IS6}$, and prove soundness and complexness theorem  with
respect to $S5$ Beth-Kripke models~\footnote{A Beth-Kripke model
is $S5$ whenever the relations $\sim_i$ are reflexive, transitive,
and Euclidian.}. The logic of $\mathbf{IS6}$ is much similar to
the $S5$ epistemic logic introduced in~\cite{kn:dit3},
pages~26-29.

 The constructive
epistemic logic $IS6$ consists of axioms $A1-A6$ and the
derivation rules $R1$ and $R2$ given below
\[
\begin{array}{l}\emph{R1:~}
\vdash\varphi,\ \vdash\varphi\rightarrow\psi\Rightarrow\ \vdash\psi\\
\emph{R2:~}\vdash\varphi\Rightarrow
K_i\varphi,\emph{~for~all}~i\in A
\quad\quad\quad\quad\quad\quad\quad\quad\quad\quad\quad\quad\quad\quad\quad\quad\quad\quad\quad\quad\quad\quad\quad\quad
\end{array} \]
\[
\begin{array}{l}\emph{A1:~ Axioms~of~constructive~propositional~logic}\quad\quad\quad\quad\quad
\quad\quad\quad\quad\quad\quad\quad\quad\quad\quad\quad\quad\\
\emph{A2:~} (K_i\varphi\wedge
K_i(\varphi\rightarrow\psi))\rightarrow K_i\psi\\
\emph{A3:~} K_i\varphi\rightarrow\varphi\\
\emph{A4:~} K_i\varphi\rightarrow K_iK_i\varphi\\
\emph{A5:~} \neg K_i\varphi\rightarrow K_i\neg K_i\varphi\\
\emph{A6:~}\neg K_i\varphi\vee K_i\varphi
\end{array} \]

\begin{theorem} (Soundness and Completeness). Axiom system $IS6$
is sound and complete with respect to semantic class of $S5$
Beth-Kripke models.
\end{theorem}\begin{proof} \end{proof}

\section{A Constructive Epistemic Public Announcement
Logic}\label{IPAL}
 We extend the proposed
 logic by adding public announcement operator to construct a
 constructive epistemic public  announcement logic,
 $\mathbf{CEPAL}$ similar  to the
(classic) public announcement  logic~\cite{kn:dit3}, page~73,

 \subsection{Syntax}
Given a finite set of agents $\mathcal{A}$, and a set of atomic
formulas $\mathbf{AT}$,  the language of
$L_{IEPAL}(\mathcal{A},\mathbf{AT})$ is inductively defined by the
BNF:

\begin{center}
$\varphi::= p \mid \neg\varphi \mid \varphi\wedge \psi,
\varphi\vee \psi \mid \varphi\rightarrow\psi \mid K_i\varphi \mid
[\varphi]\psi$
\end{center} where $i\in \mathcal{A}$, and $p\in \mathbf{AT}$.

 \subsection{Semantics}
 We    describe  public announcement operator on Beth-Kripke models.
 Let $M=\langle S,(\sim_i)_{i\in \mathcal{A}}\rangle$ be a Beth-Kripke
 model. Let $(\Theta_s,\alpha_s)$ be a Beth model of the model
 $M$. For a formula $\varphi$, we define the Beth model
 $\Theta_s|_\varphi$ as follows:

 \begin{itemize}
\item $\Theta_s|_\varphi=\{\alpha\in \Theta_s \mid
\alpha\not\Vdash_M \neg\varphi\}$,

\item the partial order relation of the Beth model
$\Theta_s|_\varphi$ is obtain by restriction of partial order
relation of the model $\Theta_s$ to the set of nodes in
$\Theta_s|_\varphi$.
 \end{itemize}
We let $M|_\varphi=\langle S',(\sim'_i)_{i\in \mathcal{A}}\rangle$
with \begin{itemize} \item[]
$S'=\{(\Theta_s|_\varphi,\alpha_s)\mid s\in S~ \&~
(M,s)\not\models \neg \varphi\}$, and \item[] $\sim'_i=\sim_i\cap
(S'\times S')$.
\end{itemize}

For each $\beta\in \Theta_s$, we define $\beta\Vdash_M
[\varphi]\psi$, if and only if for $\alpha_s$ (the root of
$\Theta_s$) we have $\alpha_s\not\Vdash_M \neg\varphi$ and
$(\Theta_s|_\varphi,\alpha_s)\Vdash_{M|_\varphi} \psi$.
 We then define
 \begin{center}
$(M,s)\models [\varphi]\psi$ iff $(M,s)\not\models \neg \varphi$
implies $(M|_\varphi,s)\models \psi$
 \end{center} Where

The dual of $[\varphi]$ is $\langle\varphi\rangle$:
 \begin{center}
$(M,s)\models \langle\varphi\rangle\psi$ iff $(M,s)\not\models
\neg \varphi$ and $(M|_\varphi,s)\models \psi$
 \end{center}

\begin{proposition}
Let $M=\langle S,(\sim_i)_{i\in \mathcal{A}}\rangle$ be a
Beth-Kripke, and $(\Theta,\alpha)$ be a Beth model of the model
$M$. For each formula $\varphi$, for each $\beta\in
\Theta|_{\varphi}$, we have $\beta\Vdash|_(M|\varphi) \varphi$.
\end{proposition}\begin{proof} It is straightforward.\end{proof}

Announcing  a formula $\varphi$,   two kinds of updates happen on
a Beth-Kripke model
\begin{itemize}
\item[I.] Updating each possible world (Beth model) of the model:

 In a constructive  possible world (a Beth model
 $(\Theta,\alpha)$),
some matters of the world are predetermined and some others are
not. When a formula $\varphi$ is announced,   all nodes of the
Beth model $\Theta$, say $\gamma\in \Theta$, which it is
impossible that  $\varphi$ gets determined in future (i.e., for
all $\beta\geq \gamma$, $\beta\not\Vdash_M \varphi$, and thus
$\beta\Vdash_M \neg\varphi$) are removed.

\item[II.] Updating the indistinguishability relations of agents:

The set of epistemic possible worlds is restricted to those
possible worlds which it is possible that $\varphi$ gets
determined, and indistinguishability relations is the same
relation regarding remained  possible worlds.
\end{itemize}
 In (classic) public announcement logic,
 $PAL$ (see~\cite{kn:dit3},~chapter~4), for a classic possible world
 $(M,s)$, and a formula $\varphi$, only one of the two formulas
 $\varphi$ and $\neg\varphi$ is \emph{announcable} (an executable announcement). It is because that in
 a classic possible world either $\varphi$ is true  or
 $\neg\varphi$, and not both of them,  and to announce a formula,
the formula must be already true in the actual world. Whereas, in
our  Beth-Kripke models,  both $\varphi$ and
 $\neg\varphi$ could be \emph{announcable}! It is because that  in
 a Beth model it is possible that neither $\neg\varphi$  nor
 $\neg\neg\varphi$ are satisfied.

 For example, assume the Beth model $(\Theta,\alpha)$, where
$\Theta=\{\alpha, \beta,\gamma\}$, $F(\alpha)=\emptyset$,
$F(\beta)=\{p\}$, $F(\gamma)=\{q\}$, $\alpha\leq \beta$, and
$\alpha\leq \gamma$. Let $M$ be a Beth-Kripke model with just one
world $s$, where $\Theta_s=\Theta$, and $\alpha_s=\alpha$, and
$s\sim_i s$, for all $i\in \mathcal{A}$. Both formula $p$ and
$\neg p$ are announcable (executable announcements) in $(M,s)$
since $(M,s)\not\models p$ and $(M,s)\not\models \neg p$, and thus
$(M,s)\models \langle p\rangle \top \wedge \langle \neg p\rangle
\top$.

\noindent The Beth-Kripke model  $M|_p$ is a model with the
world $(\Theta|_p, \alpha)$ which has two nodes $\alpha$, $\beta$,
and $F(\beta)=\{p\}$, $F(\alpha)=\emptyset$. We have
$(M|_p,s)\models \neg q$, and thus $(M,s)\models [p]\neg q$.

\noindent The Beth-Kripke model  $M|_{\neg p}$ is a model with
world $(\Theta|_{\neg p}, \alpha)$ which has two nodes $\alpha$
 $\gamma$, and  $F(\alpha)=\emptyset$, $F(\gamma)=\{q\}$. We have $(M|_{\neg
p},s)\models q$, and thus $(M,s)\models [\neg p] q$.

In (classic) public announcement logic, $PAL$
(see~\cite{kn:dit3},~chapter~4), it is possible to translate
    each formula in the language of public
announcement logic  to an equivalent   formula in the language of
epistemic logic without public announcement (see proposition~4.22,
\cite{kn:dit3}).  In contrast, in our proposed
  public announcement logic, there are formulas which
are not equivalent to any formula without announcement operator.

\begin{proposition}
Let $\mathbf{AT}=\{p\}$. The formula $\langle p\rangle \top$ is
not equivalent to any propositional formula $\varphi \in
SENT(\mathbf{AT})$.
\end{proposition}\begin{proof}
Assume  $\langle p\rangle \top$ is   equivalent to a propositional
formula $\varphi \in SENT(\mathbf{AT})$. Then for every pointed
Beth model $(\Theta,\alpha)$, $(\Theta,\alpha)\Vdash \langle
p\rangle \top\leftrightarrow \varphi$. Consider the following Beth
model $(\Theta,\alpha)$, where $\Theta=\{\alpha, \beta,\gamma\}$,
$F(\alpha)=\emptyset$, $F(\beta)=\{p\}$, $F(\gamma)=\emptyset$,
$\alpha\leq \beta$, and $\alpha\leq \gamma$. We have
$\alpha\Vdash\langle p\rangle \top$, by our assumption, we have
$\alpha\Vdash \varphi$. By item $c)$ of theorem~\ref{beth1},
$\gamma\Vdash \varphi$. Since $F(\gamma)=\emptyset$, we have
$\varphi$ is a tautology, i.e., $\varphi\leftrightarrow\top$.
Hence, we have $\langle p\rangle \top$ is a tautology, i.e.,
$\langle p\rangle \top$ is satisfied in all Beth models.
Contradiction.
\end{proof}

\begin{theorem} Regarding finite $S5$ Beth-Kripke models,   the
followings is valid, for every $\varphi, \psi\in
L_{IEPAL}(\mathcal{A},\mathbf{AT})$, and $p\in \mathbf{AT}$.

\begin{center}$[\varphi]\psi \leftrightarrow (\varphi\rightarrow
\psi)$
\end{center}
%\[
%\begin{array}{l}[\varphi]\psi \leftrightarrow (\varphi\rightarrow \psi) \quad\quad\quad\quad\quad
%\quad\quad\quad\quad\quad\quad\quad\quad\quad\quad\quad\quad\\
% [\varphi]\neg \psi \leftrightarrow \\
%\end{array} \]

\end{theorem}\begin{proof} \end{proof}

\section{The Surprise Exam Paradox}\label{SEP}

We formulate the well-known Surprise Exam Paradox, $\mathbf{SEP}$,
in our constructive epistemic public announcement logic. The
$\mathbf{SEP}$ is as follows:

The teacher announces to the students:
\begin{quote}
\emph{you will have one and only one exam at 10am on one day in
the last third days of the next week (Wednesday-Thursday-Friday),
but you will not know \emph{in advance} the day of exam.}

\end{quote}
The students, using a \emph{backward argument}, reason that
\emph{there can be no exam} indeed:
\begin{itemize}
\item[-] Friday is not the day of the exam. Since if it is, we
will not have received the exam by Thursday night, and as there is
an exam in the week, then at Thursday night, we will be able to
know in advance the day of exam,

\item[-] Thursday is not the day of the exam. Since if it is, we
will not have received the exam by Wednesday night, and as there
is an exam in the week, and it is not on Friday, then at Wednesday
night, we will know in advance the day of exam is on on
Thursday,...

\item[-] and Wednesday is not the day of the exam by a similar
argument,

\item[-] then none of the days  is the day of exam.
\end{itemize}

 The   paradox was investigated in terms of
(classic) epistemic notions by several
people~\cite{kn:Bink,kn:Quine,kn:Sore,kn:william}. For example,
Gerbrandy sees the puzzle in the assumption that announcements are
in general successful~\cite{kn:Gerb}, and Baltag, to solve the
paradox, lets the students to revise their trust to the teacher
once they reach the paradox~\cite{kn:Marc}. In~\cite{kn:RamAP},
the paradox is investigated in a constructive  view, considering
\emph{free will} of the teacher.  In this paper, we claim that

\begin{quote} the puzzle in $\mathbf{SEP}$ is that students (wrongly) assume the day that teacher is
going to take the exam is predetermined!
\end{quote}

Consider the following version of $\mathbf{SEP}$ which is
obviously equivalent to the standard version: The teacher
announces to the students:
\begin{quote}
\emph{I like  one and only one of the days  among Wednesday,
Thursday, and Friday. I start to announce one by one whether I
like the days or not beginning with Wednesday, then Thursday,  and
finally Friday. You will not know, \emph{in advance},  the day I
like before  I announce that I like that day.}
\end{quote}
Let $AT=\{p_1,p_2,p_3\}$ where
 \begin{itemize}
\item[] $p_1$ stands for "I like Wednesday",

\item[] $p_2$ stands for "I like Thursday", and

\item[] $p_3$ stands for "I like Friday".
 \end{itemize}
 It is up to desire of the teacher to like which day, and his
 desire could be non-predetermined. Therefore, the teacher can
 choose freely  announce  either $p_1$ or $\neg p_1$  (certainly, not both of
  them together). But, in a classic epistemic possible world either $p_1$ is true or
  $\neg p_1$, and thus just one of the formulas $p_1$ and $\neg
  p_1$ is an executable-announcement, and thus classic epistemic possible worlds are not
  suitable to formulate the paradox.

When the teacher uses the  term `like', he means that he announces
up to his free will, and students
  must assume that both  $p_1$ and $\neg   p_1$ are
  executable-announcements at the beginning,  and if the teacher
  announces $\neg p_1$ then both  $p_2$ and $\neg   p_2$ are
  executable-announcements, and finally if the teacher first
  announces $\neg p_1$ then announces $\neg p_2$, at the end he
  can announce $p_3$, or in other words, $p_3$ is an executable
  announcements. Therefore, more formally, if $(M,s)$ is a (constructive) epistemic possible world which represents
  the situation of the $\mathbf{SEP}$ then we have
  \begin{itemize}
\item[] $(M,s)\models \langle p_1\rangle\top \wedge \langle \neg
p_1\rangle \top \wedge \langle \neg p_1\rangle\langle
p_2\rangle\top \wedge \langle \neg p_1\rangle\langle \neg
p_2\rangle\top \wedge \langle \neg p_1\rangle\langle \neg
p_2\rangle \langle  p_3\rangle\top$.
  \end{itemize}

  Let us formally describe teacher's announcement in the language
  of our proposed logic.
\begin{itemize}
\item[Claims0] The teacher claims that he likes  one and only one
of the days among Wednesday, Thursday, and Friday:
$$\varphi_0:\equiv (p_1\vee p_2\vee p_3)\wedge \neg(p_1\wedge
p_2)\wedge \neg(p_1\wedge p_3)\wedge\neg (p_2\wedge p_3).$$

\end{itemize}
The teacher says: "I start to announce one by one whether I like
the days or not beginning by Wednesday, then Thursday,  and
finally Friday. You will not know, \emph{in advance},  the day I
like before  I announce that I like that day".

\noindent He means that at least \emph{one} of the three following
claims is true.

\begin{itemize}
\item[claim1)] I (the teacher) can announce that I like Wednesday,
and you (students) would not know that I like Wednesday before I
announce it: $$\varphi_1\equiv:\langle p_1\rangle \top \wedge \neg
K p_1.$$
 \item[claim2)] I can announce that I do not like
Wednesday, and after that I can announce that I like Thursday, but
you would not know that I like Thursday before I announce it
$$\varphi_2:\equiv\langle \neg p_1\rangle\langle p_2\rangle \top
\wedge \langle\neg p_1\rangle \neg K p_2.$$

\item[claim3)] I can announce that I do not like Wednesday, and
after that I can announce  that I do not like Thursday, and
finally I can announce that I like Friday, but you would not know
that I like Friday before I announce it  $$\varphi_3:\equiv\langle
\neg p_1\rangle\langle \neg p_2\rangle \langle p_3\rangle\top
\wedge \langle\neg p_1\rangle\langle\neg p_2\rangle \neg K p_3.$$
\end{itemize}
In this way the teacher's announcement is
\begin{center}
$\varphi_0\wedge (\varphi_1\vee\varphi_2\vee\varphi_3)$.
\end{center}

 The following Beth-Kripke model $M=\langle
S,(\sim_i)_{i\in \mathcal{A}}\rangle$ represents  the The
(constructive) epistemic world of $\mathbf{SEP}$, with
\begin{itemize}
\item[] $S=\{s\}$,

\item[] $\mathcal{A}=\{student\}$,

\item[] $\sim_{student}=\{(s,s)\}$,
\end{itemize}where $(\Theta_s,\alpha_s)$ is defined as follows:
$\Theta_s=\{\alpha_s,\beta,\gamma,\delta\}$,
$F(\alpha_s)=\emptyset$, $F(\beta)=\{p_1\}$, $F(\gamma)=\{p_2\}$,
$F(\delta)=\{p_3\}$, and $\alpha_s\leq \beta$, $\alpha_s\leq
\gamma$, $\alpha_s\leq \delta$.

\noindent One may easily verify that $(M,s)\models \varphi_0\wedge
(\varphi_1\vee\varphi_2\vee\varphi_3)$.

 The students   reasons
that the third claim, $\varphi_3$, is not true as the teacher
likes at least one of the days, and after he announced that he
does not like Wednesday and Thursday, the student deduce that he
must like Friday. However, students cannot apply a backward
argument in constructive epistemic logic.

The student correctly argues that the assumption of $\varphi_3$
leads to contradiction, and we also have $(M,s)\models \neg
\varphi_3$. Then they again correctly  derives from $(M,s)\models
\neg \varphi_3$ that $(M,s)\not\models p_3$.

\noindent In classic epistemic logic negation of a formula is
defined by non-satisfaction of the formula (i.e., $((N,t)\models
\neg\psi$ iff $(N,t)\not\models \psi$), thus students derive
$(M,s)\models \neg p_3$, and using $\neg p_3$, by a backward
argument, they
    derive $\neg p_2$ , and after that $\neg
 p_1$.

\noindent  However, in Beth-Kripke model, we \emph{cannot} derive
from $(M,s)\not\models
 p_3$ that $(M,s)\models \neg p_3$.  The atomic formula $p_3$ is
 not predetermined at the  node $(\Theta_s,\alpha_s)$, and we have
 $\alpha_s\not\Vdash p_3$, and $\alpha_s\not\Vdash \neg p_3$. In this way, the
 student's reasoning does not work in constructive epistemic
 logic with public announcement.


\begin{thebibliography}{10}

\bibitem{kn:RamAP} M. Ardeshir, R. Ramezanian, {\em A solution to the Surprise Exam Paradox in Constructive Mathematics},
Journal of Review of Symbolic logic, 2012.



\bibitem{kn:Bink} R. Binkley, {\em The Surprise Examination in Modal Logic}, J. Phil, 1968.

%\bibitem{kn:Chow} T. Y. Chow, {\em The Surprise Examination or Unexpected Hanging
%Paradox}, Amer. Math. Monthly, 1998.

\bibitem{kn:dalen}  D. van Dalen, {\em  An interpreatation of intuitionistic
analysis}, Annals of Mathematical Logic, 13, 1-43, 1978.

%\bibitem{kn:vandit} H. van Ditmarsch, B. Kooi,
%{\em The secret of my success}, Synthese 151, 2006.

\bibitem{kn:dit3} H. van Ditmarsch, W. van der Hoek, and B. Kooi,
{\bf Dynamic Epistemic Logic}, Springer, 2008.



\bibitem{kn:Gerb} J. Gerbrandy, {\em The  Surprise Examination in Dynamic Epsitemic
Logic}, Synthese, 2007.

\bibitem{kn:Kapl} W. Kaplan, R. Montague, {\em A paradox Regained}, Notre Dame J. Formal
Logic 1, 1960.

%\bibitem{kn:Krit} S. Kritchman, R. Raz, {\em The Surprise
%Examination Paradox and the Second Incompleteness Theorem},
%Notices of the AMS, vol 57, 2010.

%\bibitem{kn:Levy} K. Levy, {\em The solution to the Surprise Exam Paradox},
%The Southern Journal of Philosophy, 2009.

\bibitem{kn:Marc} A. Marcoci, {\em The Surprise Examination Paradox, A review of
two so called solutions in dynamic epistemic logic}, 2011.

\bibitem{kn:Conner} D. O'Conner, {\em Pragmatic Paradoxes}, Mind 57,
1948.

%\bibitem{kn:Pett} P. Pettit, R. Sugden {\em The Backward Induction Paradoxes}, The
%Journal of Philosophy, Vol 86, 1989.

\bibitem{kn:Quine} W. Quine, {\em On a so-called Paradox}, Mind 62,
1953.

\bibitem{kn:Shaw} R. Shaw, {\em The Paradox of the Unexpected Examination}, Mind 67,
1958.

%\bibitem{kn:Sobe} J. H. Sobel, {\em Backward Induction Arguments: A Paradox Regained},
%Philosophy of Scince, 1993.

\bibitem{kn:Sore} R. A. Sorensen, {\em Conditional Blindspots and the Knowledge
Squeeze: a solution to the prediction paradox}, Australian
Journal. Phil, 1984.

\bibitem{kn:TD} A. S. Troelstra, D. van Dalen,
{\bf Constructivism in Mathematics, An introduction}, Vols. 1, 2,
 North-Holland, 1988.

%\bibitem{kn:Weih} K. Weihrauch,
%{\bf Computable Analysis, An introduction}, Springer, 2000.

%\bibitem{kn:Wiss} P. Weiss,
%{\em The prediction paradox}, Mind 61 1952.

\bibitem{kn:william} T. Williamson,
{\bf Knowledge and its limits}, Oxford University Press, 2002.

\end{thebibliography}
\end{document}